\documentclass[12pt,reqno]{amsart}
\usepackage[samelinetheorem,notitle]{maherart}

\usepackage{xurl}

\usepackage{amsfonts,enumerate,bm,bbm,paralist,nicefrac,comment,float,thmtools,thm-restate,comment,multirow}
\usepackage[export]{adjustbox}
\makeatletter
\newcommand\fs@boxedtop
  {\fs@boxed
   \def\@fs@mid{\vspace\abovecaptionskip\relax}%
   \let\@fs@iftopcapt\iftrue
  }
\makeatother
\floatstyle{boxedtop}
\floatname{framedbox}{Procedure}
\newfloat{framedbox}{tbp}{lob}

\hypersetup{
    pdftitle =
        {Censorship Resistance in On-Chain Auctions},
    pdfauthor =
        {Elijah Fox, Mallesh M Pai and Max Resnick},
    pdfsubject=
        {Censorship Resistance in On-Chain Auctions}
}
\renewcommand{\R}{\mathbb{R}}

\usepackage{xparse}

\DeclareMathOperator{\writeop}{write}
\DeclareMathOperator{\readop}{read}

\DeclareDocumentCommand\Pr{ m g }{%
    \ensuremath{   \IfNoValueTF {#2}
      {\mathbb{P}\left[{#1}\right]}
      {\mathbb{P}\left[{#1}\middle\vert{#2}\right]}%
    }
}
\DeclareDocumentCommand\E{ m g }{%
    \ensuremath{   \IfNoValueTF {#2}
      {\mathbb{E}\left[{#1}\right]}
      {\mathbb{E}\left[{#1}\middle\vert{#2}\right]}%
    }
}

\renewcommand\thmcontinues[1]{Continued}

\newif\ifdraft
\drafttrue

\begin{document}
\raggedbottom
\title{Censorship Resistance in On-Chain Auctions}
\author[Fox]{Elijah Fox$^\text{A}$}
\author[Pai]{Mallesh M. Pai$^\text{B,C}$}
\author[Resnick]{Max Resnick$^\text{D}$}
\address{$^\text{\MakeLowercase{A}}$Duality Labs
\href{mailto:elijah@duality.xyz}{elijah@duality.xyz}}
\address{$^\text{\MakeLowercase{B}}$Department of Economics, Rice University %
\href{mailto:mallesh.pai@rice.edu}{mallesh.pai@rice.edu}}
\address{$^\text{\MakeLowercase{C}}$Special Mechanisms Group
\href{mailto:mallesh.pai@specialmechanisms.org}{mallesh.pai@specialmechanisms.org}}
\address{$^\text{\MakeLowercase{D}}$Special Mechanisms Group
\href{mailto:max@specialmechanisms.org}{max@specialmechanisms.org}}

\address{\today}
\thanks{Pai gratefully acknowledges support from the NSF (CCF-1763349).}
\begin{abstract}
Modern blockchains guarantee that submitted transactions will be included eventually; a property formally known as liveness. But financial activity requires transactions to be included in a timely manner. Unfortunately, classical liveness is not strong enough to guarantee this, particularly in the presence of a motivated adversary who benefits from censoring transactions. We define \emph{censorship resistance} as the amount it would cost the adversary to censor a transaction for a fixed interval of time as a function of the associated tip. This definition has two advantages, first it captures the fact that transactions with a higher miner tip can be more costly to censor, and therefore are more likely to swiftly make their way onto the chain. Second, it applies to a finite time window, so it can be used to assess whether a blockchain is capable of hosting financial activity that relies on timely inclusion.  

We apply this definition in the context of auctions. Auctions are a building block for many financial applications, and censoring competing bids offers an easy-to-model motivation for our adversary. Traditional proof-of-stake blockchains have poor enough censorship resistance that it is difficult to retain the integrity of an auction when bids can only be submitted in a single block. As the number of bidders $n$ in a single block auction increases, the probability that the winner is not the adversary, and the economic efficiency of the auction, both decrease faster than $1/n$. Running the auction over multiple blocks, each with a different proposer, alleviates the problem only if the number of blocks grows faster than the number of bidders. We argue that blockchains with more than one concurrent proposer have can have strong censorship resistance. We achieve this by setting up a prisoner's dilemma among the proposers using conditional tips. 



\end{abstract}


\maketitle

\thispagestyle{empty}

\newpage

\pagenumbering{arabic}

\section{Introduction}
Blockchain consensus algorithms typically guarantee \emph{liveness},  meaning valid transactions will be included on chain eventually. But financial applications, are time sensitive. For these  to function as intended, valid transactions must be included on the blockchain in a timely manner. This requires something stronger than liveness: \emph{censorship resistance}. To quote \cite{buterin2015censorship}, censorship resistance is ``ensuring that transactions that people want to put into the blockchain will actually get in in a timely fashion, even if ``the powers that be'', at least on that particular blockchain, would prefer otherwise.'' 

In this paper, we propose a formal definition that quantifies censorship resistance in the sense of \cite{buterin2015censorship} above. We abstract away from the details of the chain and view it as a public bulletin board with two operations: a read operation, which always succeeds, and a write operation. The write operation succeeds when the transaction with an associated tip is added to the bulletin board and fails otherwise. We then define the \emph{censorship resistance} of this public bulletin board as the amount it would cost a motivated adversary to cause a write operation to fail, as a function of the associated tip.

This definition has two advantages that are useful in applications. The first stems from the fact that we define censorship resistance as a function of the underlying tip: this captures how transactions with a higher tip can be more expensive to censor, and therefore are more likely to successfully make it on chain. The second is that by capturing the censorship resistance of a specific public bulletin board, with an associated length of time for a transaction to be added/ considered censored, we can capture the trade-off between censorship resistance an speed for specific blockchain designs. 


Having presented this definition, we apply it to tackle whether a given public bulletin board is sufficiently censorship resistant to host a given mechanism. This is tricky because it depends on both the tips of the underlying transaction(s) \emph{and} the motivation of the adversary, both of which are potentially endogenously determined by the mechanism.  

This paper considers whether existing proof of stake blockchains are sufficiently censorship resistant to host time sensitive auctions. We consider auctions for two reasons. First, the cost of being censored and the benefits of censoring competing bids are easy to quantify in an auction. Therefore, for a given public bulletin board with fixed censorship resistance, we can determine whether the auction will actually function as intended. Second, auctions are already a popular mechanism on-chain: for example, Maker DAO, the entity behind popular stablecoin DAI, uses Dutch clock auctions to sell the right to liquidate collateral for distressed loans, and additionally digital goods such as NFTs may also be auctioned off on-chain. Also, several important future developments will require auctions on-chain, from the very organization of Ethereum (\cite{monnot2022pepc}), to proposals by other large organizations to move current off-chain auctions on-chain (\cite{flashbots_2022}).   

Formally, we consider a seller of a single unit of an indivisible good who runs a second-price auction on-chain. The seller encodes the rules of the second-price auction in a smart contract. The contract selects the bidder who submitted the highest bid over a predefined period and sets the payment to the second-highest bid. But since all of this takes place on a blockchain, before a bid can be submitted to the auction, it must be included in a transaction on-chain. Valid transactions are submitted to a \emph{mempool}. Each slot, the proposer gathers transactions from the mempool into a block that will eventually be added to the chain. Proposers have complete autonomy over which transactions to include. 

The power of the proposer to determine the contents of the block and, therefore, the outcome of the auction sets up a competition for inclusion. Bidders include tips for the proposer along with their bids. The proposer receives these tips if and only if the corresponding transactions are included in his block. We suppose that there is a single colluding bidder who may offer a bribe to the proposer in exchange for omitting certain transactions. In distributed systems terminology, this single colluding bidder who may bribe the proposer if they believe it profitable for them to do so is our \emph{threat model}.

We show that in the auction setting, tips for inclusion are a public good since they provide security to other transactions and only benefit the bidder who pays the tip if they win the auction. In contrast, bribes for omission are purely for private benefit and make other bidders worse off in equilibrium. Consequently, our results suggest that the colluding bidder is highly advantaged in this game. In particular, we show that as the number of honest bidders increases, the colluding bidder wins the auction increasingly often and collects an increasingly large share of the surplus created by the auction. 

Notably, we assume that bids are sealed, that is, that the colluding bidder must choose which transactions to try and omit based solely on the associated tip. We show that the colluding bidder can back out the private bids (and therefore whether it is profitable to attempt to censor these bids) based on these public tips. This suggests that cryptographic approaches (e.g. commit-reveal schemes) are not a silver bullet for resolving censorship concerns in such settings. 

We then consider two alternate designs that improve censorship resistance. The first is to run the auction over multiple slots with a different proposer for each. We find that this achieves sufficient censorship resistance only if the number of blocks grows faster than the number of bidders. This is undesirable for reasons external to our model; for example, executing the auction in a short window is important for financial applications. MEV auctions in particular are concerned about speed, since they need to clear at least once every slot -- once every 12 seconds on Ethereum. 

The second is to have blockchains with multiple concurrent block proposers, $k>1$, and allow bidder tips to be conditioned not only on inclusion, but also on the number of proposers who include the bid within a slot. This allows bidders to set up a sort of ``prisoner's dilemma'' among the proposers by offering to pay a large tip $T$ when only one proposer includes, and a small tip $t \ll T$ if multiple proposers include.  Each proposer is incentivized to include since if they are alone in including the transaction, there is a high tip attached. Therefore, all proposers include the bid in equilibrium. However, censoring is expensive, since censoring requires that each proposer be bribed $T$ for a total cost of $kT.$ This leads to a low expected tip of $kt$ but an asymmetrically expensive censorship cost of $kT \gg kt$. This asymmetry allows for a pooling equilibrium in which the probability of censorship is 0, the tips no longer reveal the bids, and the expected total tips are low.

\section{Related Literature}

Censorship resistance, for various definitions of the term, is a key desideratum motivating the adoption of blockchains. This property was mentioned in some of the earliest writings on the subject, e.g.,  \cite{buterin2015censorship}. More recently, this property has come under additional scrutiny due to two major developments. The first, Proposer Builder Separation (PBS) in Ethereum, explicitly establishes an auction for the right to build the next block. Block builders who win this auction decide which transactions make it onto the chain and, more importantly for our purpose, which transactions do not. PBS therefore enables a motivated adversary to censor specific transaction(s) by purchasing the right to build the next block and intentionally omitting those transactions (\cite{buterin2021pbs}). The second relates to US OFAC sanctions on certain Ethereum addresses, and the subsequent decision by certain block builders to exclude transactions including those addresses from blocks that they build. Effects of this (and a related definition of censorship resistance) are studied in \cite{wahrstatter2023blockchain}. 

The literature on auctions, even restricting to papers that explicitly think about auctions in the context of blockchains, is much larger. In such auctions, bids are rarely announced simultaneously, and maintaining the seal on bids transmitted through public channels requires cryptography. For example, a simple cryptographic second-price sealed-bid auction involves bidders submitting the hash of their bids rather than the bids themselves and then revealing the hash after all bids have been submitted. \cite{ferreira2020credible} showed that, using a cryptographically secure commitment scheme, it is possible to design an auction that is optimal, strategy proof and credible (in the sense of \cite{akbarpour2018credible,akbarpour2020credible}). More complicated cryptographic approaches can eliminate the need to reveal any information beyond the results of the auction and can also accommodate combinatorial auctions (\cite{lee2001efficient,elkind2004interleaving,suzuki2003secure}).

Auctions are commonly cited as use cases for the verifiable computation that smart contracts provide. This was originally envisaged in \cite{szabo1997formalizing}, who noted that ``$\ldots$ a blockchain with a built-in fully fledged Turing-complete programming language that can be used to create ``contracts'' $\ldots$ simply by writing up the logic in a few lines of code'', see also \cite{galal2019verifiable, blass2018strain}. Auctions have also been suggested as a desirable mechanism to decide the order and inclusion of transactions to mitigate MEV (\cite{kulkarni2022towards}). Historically, these were decided by a combination of auction and speed-based mechanisms, leading \cite{daian2020flash} to compare MEV with high-frequency trading as described in \cite{budish2015high}. Initially, inclusion and priority within the block were decided by priority gas auctions (PGAs), since most validators gathered transactions directly from the mempool and ordered them according to their miner tips, breaking ties using a first come first serve rule (\cite{daian2020flash}). But recently, a super-majority of validators have switched their execution clients to MEV-boost compatible versions, meaning the right to decide inclusion and ordering for most blocks is sold to the highest bidder. These bidders are typically established builders who specialize in extracting the maximum value from each block. The leading advocate for this approach has been Flashbots, the company behind the initial open source MEV client. Their next product SUAVE, aims to move these auctions on-chain \cite{flashbots_2022}.

Previous MEV mitigation research has focused on fairness rather than censorship (\cite{kelkar2020order, kelkar2021themis}). But \cite{ferreira2022credible} showed that for every sequencing rule of trades through a liquidity pool, there exists a way for the proposer to obtain non-zero risk-free profits suggesting that ordering based MEV is inevitable with current on chain financial application design. In response to this, researchers have suggested that frequent batch auctions or other order-agnostic mechanisms might alleviate the MEV that arises from transaction ordering power (\cite{johnson2022concave}). 

On-chain auctions have also been studied as a mechanism for the sale of non-fungible tokens (NFTs) \cite{milionis2022framework}. Gradual dutch auctions (GDAs) \cite{frankie_robinson_white_andy8052_2022}, are a dynamic mechanism for selling multiple NFTs. \cite{kulkarni2023credible} explores the credibility of GDAs and finds that an auctioneer can bid to artificially raise the sale price and create the appearance of demand.
\section{Formalizing Censorship Resistance}

As we described above, we abstract away from the details of the blockchain and instead consider solely the functionality as a public bulletin board. For instance, for our leading application of the auction , the public bulletin board can be written to, which is how bids may be submitted, and can be read from, which is how the auction can then be executed. For simplicity, we assume that once a message has been successfully written to , it can be read without friction. For example, any transaction included in an Ethereum block can be read by any full node.

\begin{definition}[Public Bulletin Board]
A Public Bulletin Board has two publicly callable functions:
\begin{enumerate}
    \item $\writeop(m,t)$ takes as input a message $m$ and an inclusion tip $t$ and returns $1$ if the message is successfully written to the bulletin board and $0$ otherwise.
    \item $\readop()$ returns a list of all messages that have been written to the bulletin board over the period. 
\end{enumerate}
\end{definition}

Some subtlety must be observed in the definition of a public bulletin board. First, the $\writeop$ function takes as input not only a message $m$ but also a tip $t$. Here our definition departs from \cite{choudhuri2017fairness}. This models proposer tips and other forms of validator bribes. To motivate this, consider that the write function on Ethereum is unlikely to succeed without a sufficient tip, and so the behavior of $\writeop$ is very different depending on the size of the associated tip.

\theoremstyle{definition}
\newtheorem{myexp}{Example}

\begin{myexp}[Single Block] \label{ex:single}
    In the case of a single block, the $\writeop(m,t)$ operation consists of submitting a transaction $m$ with associated tip $t$. $\writeop(m,t)$ succeeds if the transaction is included in that block on the canonical chain. 
\end{myexp}

\begin{myexp}[Multiple Blocks] \label{ex:multiple}
    In the case of $k$ blocks with rotating proposers, the $writeop(m,t)$ operation consists of submitting a transaction $m$ with associated tip $t$ during the period before the first block is formed. $\writeop(m,t)$ succeeds if the transaction is included in any of the $k$ blocks.
\end{myexp}

We can now model the relationship between the tip $t$ and the success of the write operation as a function. This function provides a flexible definition of the bulletin board's censorship resistance.

\begin{definition}[Censorship Resistance of a Public Bulletin Board]
The censorship resistance of a public bulletin board $\mathcal{D}$ is a mapping $\phi: \R_+ \to \R_+$ that takes as input the tip $t$ corresponding to the tip in the write operation $\writeop(\cdot,t)$ and outputs the minimum cost that a motivated adversary would have to pay to make the $\writeop$ fail. 
\end{definition}

This definition allows us to compare the censorship resistance of two bulletin boards even when the inner machinations of those mechanisms are profoundly different. This definition also easily extends to cases where tips are multi-dimensional. i.e. $t \in \R_+^n$ by substituting $\phi: \R_+ \to \R_+$ to $\phi: \R_+^n \to \R_+$.  

\renewcommand{\thmcontinues}[1]{continued}

\begin{myexp}[continues=ex:single]


The cost to censor this transaction would simply be $t$ in the uncongested case because the motivated adversary has to compensate the proposer at least as much as the proposer would be losing from the tip. So:
\begin{equation}
    \phi(t) = t
\end{equation}
\end{myexp}

\begin{myexp}[continues=ex:multiple]

In the case of $k$ blocks with rotating proposers, the $\writeop(m,t)$ operation consists of submitting a transaction $m$ with associated tip $t$ during the period before the first block is formed. $\writeop(m,t)$ succeeds if the transaction is included in any of the $k$ blocks. The cost to censor this transaction would be $kt$ since each of the $k$ proposers must be bribed at least $t$ to compensate them for the forgone tip on the transaction that they each had the opportunity to include. So:
\begin{equation}
    \phi(t) = kt.
\end{equation}
\end{myexp}

\section{Modelling the Auction}
We are now in a position to use our definitions to understand the censorship resistance of a sealed-bid auction, for various design parameters. 

We consider a traditional independent private values setting. There is a single seller with a single unit of an indivisible good for sale. There are $n+1$ buyers for the good, $n\geq 1$--- we denote the set of bidders by $N= \{0,1,\ldots,n\}$. Each of these buyers $i \in N$ has a private value for the good, $v_i$ . We suppose buyer 0's value $v_0$ is drawn from a distribution with CDF $F_0$  and density $f$, and the other buyers' values are are drawn i.i.d. from a  distribution with CDF $F$ and density $f$. Both distributions have bounded support, we normalize these to be equal to the unit interval $[0,1]$. Several of our results will be for the special case $F=F_0= U[0,1]$. Bidders know their own $v_i$; and $n$, $F,$ and $F_0$ are common knowledge among all bidders and the seller. 

The seller wants to conduct a sealed bid second-price auction with reserve price $r$. As described in this introduction, the point of departure of our model is that this auction runs on a blockchain. Initially, we consider an auction that accepts bids in a single designated block. Below, we formally define this game and our solution concept. 

In an idealized world with honest/ non-strategic proposers, the auction would run as follows: 
\begin{enumerate}
\item The seller announces the auction. 
\item All buyers privately commit their bids to the auction as transactions. 
\item Proposer(s) include these transactions on  relevant block(s).\footnote{We abstract away from issues such as block size constraints/ congestion.} 
\item The second-price auction is computed based on the included bids, i.e., the highest bid is selected to win if this bid $\geq r$, in this case paying a price of $\max\{r, \textrm{other bids}\}$. 
\end{enumerate}

In particular, we assume that the idealized sealed-bid nature of off-chain auctions can be achieved on-chain via cryptographic methods\footnote{This can be practically achieved by submitting the hash of the bid and revealing it later.}. We also assume that the set of bids submitted for the auction is public (the bid itself may be private, but the fact that it exists as a bid is public). 

Our main concern is that bids submitted for this auction may be censored, that is, omitted from a block. More specifically, we suppose that after all other bids are submitted, but before they are revealed, a designated bidder, bidder 0, can pay the proposer of the block to \emph{censor} bids. These censored bids are then excluded from the block and have no impact on the auction. We assume that the proposer is purely profit focused and that bidder's utilities are quasilinear. 

Formally, we consider the following game:
\begin{enumerate}
    \item The seller announces second-price sealed-bid auction with reserve price $r$ to be conducted over a single block. 
    \item Buyers learn their values $v_i \leftarrow F$. 
    \item Buyers $1, \ldots, n$ each simultaneously submit a private bid $b_i$ and a public tip $t_i$. 
    \item Buyer $0$ observes all the other tips $t_i$ and can offer the proposer of the block a take-it-or-leave-it-offer of a subset $S \subseteq \{1,\ldots, n\}$ of bidders and a bribe $p$ to exclude that subset's bids. Bidder $0$ also submits his own bid $b_0$.
    \item The proposer accepts or rejects bidder $0$'s bribe and constructs the block accordingly, either including $N \setminus S$ if he accepts or $N$ otherwise. 
    \item The auction is computed based on the bids included in the block.
\end{enumerate}
In the next section, we consider the case of auctions over multiple sequential blocks with independent proposers and the case of simultaneous proposers. Those games are variants of the game above. We describe them in-line. 

Formally, pure strategies in this game are:
\begin{itemize}
\item For players $i \in \{1,\ldots, n\}$: A tuple of bidding and tipping strategies $\beta_i: [0,1] \rightarrow \Re_+$, $\tau_i: [0,1] \rightarrow \Re_+$ for players $i \in \{1,\ldots n\}$, that is, player $i$ with value $v_i$ bids $\beta_i(v_i)$ and tips $\tau_i(v_i)$. 
\item For player $0$: an offer to the proposer $\theta_0: \Re_+^n \times [0,1] \rightarrow 2^N \times \Re_+$, and a bid function $\beta_0 : 
\Re_+^n \times [0,1] \rightarrow \Re_+$, i.e. as a function of tips $\mathbf{t} = (\tau_1(v_1),\dots,\tau_n(v_n))$ and his own value $v_0$, an offer $\theta_0(\mathbf{t},v_0)$ and a bid $\beta_0(\mathbf{t},v_0)$. 
\item For the proposer, given the tips $\mathbf{t}$ and an offer from player $0$, $\theta_0(\mathbf{t},v_0)$, a choice of which bids to include.
\end{itemize}

Since our game is an extensive-form game of incomplete information, our solution concept is the Perfect Bayes-Nash Equilibrium (PBE). This requires strategies to be mutual best-responses, as is standard in most equilibria. Additionally, for each player, it requires the player to have beliefs about unknowns at every information set (on-, and off-, path) at which they are called upon to play such that their strategy maximizes their expected utility given the beliefs and others' strategies. Beliefs are correct on path (i.e., derived from the prior, and Bayesian updating given agents' strategies), and unrestricted off path.

Note that the proposer has multiple potential indifferences: e.g., should they include a bid with 0 tip? Should they censor a set of bids if bidder $0$'s offered bribe exactly equals the total tip from that set?  We assume that given tips $\mathbf{t}$ from bidders $\{1,\ldots,n\}$ and an offer from bidder $0$ to censor subset $S$ for a bribe of $p$, the proposer includes the bids of $N-S$ if and only if $p \geq \sum_{i \in S} t_i$, and includes bids from all $N$ otherwise (i.e., we break proposer indifferences in favor of bidder 0 so that best responses are well defined). 

There are multiple PBEs of the game, driven in part by the fact that there are multiple equilibria in a second price auction (for instance, there is an equilibrium in the second price auction for one player to bid a high value and all others to bid $0$). However, most of these equilibria are in weakly dominated strategies. We therefore focus on the following class of equilibria:
\begin{enumerate}
    \item Bidders $\{1,\ldots n\}$ submit a truthful bid, i.e. $\beta_i(v_i) = v_i$. Note that this is a weakly dominant strategy for them. In addition, these bidders use a symmetric tipping function $\tau$, that is, $\tau_i(\cdot) = \tau(\cdot)$.
    \item Bidder $0$ bids equal to his value if he believes, given the tips of $\{1,\ldots n\}$, that there is a nonzero probability that he could win the auction, otherwise he bids $0$ or does not bid. 
\end{enumerate}
In what follows, we simply refer to a PBE that satisfies this refinement as an \emph{equilibrium} of the game (with no qualifier). We reiterate that there are multiple PBEs of the original game; we are simply restricting attention to these ``reasonable'' equilibria for tractability. 

\section{Results}
Our results are easiest for the case that n = 1, that is, there are two bidders. We present this as an illustration before considering the general case.

\subsection{Two Bidder Case}
Suppose there are only two bidders, one ``honest'' bidder 1 with value drawn according to distribution $F$, and one ``colluding'' bidder who has the opportunity to collude with the proposer, bidder 0, with value drawn independently from distribution $F_0$. We assume that $F_0$ satisfies a regularity condition, that is, that $F_0(t)/f_0(t)$ is non-decreasing in $t$. 

The equilibrium in this case is easy to describe:
\begin{proposition}\label{prop:eqbm1}
The following constitutes an equilibrium of the game with 2 bidders, i.e. $N= \{0,1\}$, when the seller announces a second-price auction with a reserve price $r=0$: 
\begin{itemize}
    \item Bidder 1 submits a truthful bid, and his tipping strategy as function of his value $v_1$ is given by $t_1(v_1)$ solves $(v_1 - t) -\frac{F_0(t)}{f_0(t)} =0 $. 
    \item Bidder $0$'s strategy as function of the observed tip $t_1$ and his value $v_0$ is given by 
    \begin{equation*}
    \sigma_0(t_1,v_0) = \begin{cases}
        \text{bribe} & t_1 \leq v_0,\\
        \text{don't bribe} & t_1 > v_0.\\
    \end{cases}
\end{equation*}
where bribe is shorthand for paying $t_1$ to the proposer in exchange for omitting bidder 1's transaction. Further bidder $0$ submits a nonzero bid in the auction if and only if he bribes the proposer.
\item The proposer accepts bidder $0$'s bribe whenever it is offered and omits bidder $1$'s bid, otherwise the proposer includes both bids. 
\end{itemize}
\end{proposition}

Before providing a proof of this result, the following corollary summarizes the outcome that results in this equilibrium when $F = F_0 =\text{Uniform}[0,1]$. For comparison, recall that in the (standard) second price auction when both buyers have values drawn i.i.d. from $\text{Uniform}[0,1],$ the expected revenue is $1/3$ and each bidder has an \textit{ex ante} expected surplus of $1/6$. 

\begin{corollary}\label{cor:outcome1}
Bidder $0$ wins the object with probability $\tfrac34$, and has an expected surplus of $\frac{13}{48}$, while bidder $1$ wins the auction with probability $\frac14$ and has an expected surplus of $\frac{1}{12}$.  Revenue for the seller in this auction is $0$, and the expected tip revenue for the proposer is $\tfrac14$. 
\end{corollary}

In short, the proposer collects all the revenue from this auction, while the seller collects none. Bidder $0$ is substantially advantaged by his ability to see bidder $1$'s tip and then decide whether to bribe the proposer or not (wins the auction with higher probability, collects more of the surplus). Our results for $n>1$ are similarly stark except for the fact that the auctioneer collects some positive revenue when $n >1$, although this revenue rapidly decreases as $n$ increases. 

\begin{proof} The proof of the proposition is straightforward so we describe it briefly in line. First to see that bidder $1$ should bid his value (our refinement restricts attention to these) note that bidder 1 has two actions, he privately submits a bid $b_1$ and publicly submits a tip $t$.  Since bids only matter after inclusion has been decided, which is also after tips have been paid, tips are a sunk cost and what remains is simply a second price auction, in which truthful bidding is a weakly dominant strategy. Therefore it is (part of) an equilibrium for bidder $1$ to bid his value.


Notice that bidder 1 does not benefit from submitting a tip $t > v_1$ since even if he wins, he will end up paying more in tips (in addition to possible fees from the auction) than he values the item. Knowing this, it is always weakly better for player 0 to pay $t$ to omit player 1's bid when $v_0 \geq t$. Thus player 0 bribes the proposer if and only if $v_0 \geq t$. In that case player $1$'s expected utility as a function of his tip is:
\begin{align*}
    \mathbb{E}[U_1(v_1,t)] &= F_0(t) \left(v_1 - t\right).
\end{align*}
Taking the derivative with respect to $t$, and setting it equal to 0, we get, as desired, that
\begin{align*}
    t_1(v_1) \; \text{solves} \; (v_1 - t) -\frac{F_0(t)}{f_0(t)} =0 .
\end{align*}
Now that we have found a candidate equilibrium, we have some more work to do to verify that it is in fact a PBE. Formally, bidder 1's beliefs at his only information set are that $v_0 \sim \text{Uniform}[0,1]$. By our regularity condition, $t_1(\cdot)$ is strictly increasing. 
Bidder 0's beliefs, conditional on bidder 1's tip being $t_1$ are given by
\begin{equation*}
    v_1 = \begin{cases}
        t_1^{-1}(t_1) & t_1 \leq t_1(1),\\
        1 & \text{ otherwise.}
    \end{cases}
\end{equation*}
Notice that the case of $t_1 > t_1(1)$ is off the equilibrium path. Finally note that Bidder $0$'s strategy to bribe whenever his value exceeds bidder $1$'s tip, and to bid only if he is willing to bribe, constitutes a best response. This concludes the proof. 
\end{proof}

Notice that even though bids are completely private in this model, because of the transaction inclusion micro-structure, bids are effectively revealed by the tips attached to them. This calls into question whether we can conduct sealed bid auctions of any type on chain.  

We can also describe the equilibrium for the case where the seller chooses an auction with a reserve price $r>0$. For brevity we describe this informally:
\begin{align*}
&t_1(v_1)\; \text{solves}\;  (v_1 - r - t) f_0(r+t) - F_0(r+t) =0 \text{ if solution exists,}\\
&t_1(v_1) = 0\; \text{otherwise}. 
\end{align*}
Our regularity condition implies that there exists $\underline{v} = r + \frac{F_0(r)}{f_0(r)}> r$ such that $t_1(v_1) = 0$ for $v \leq \underline{v}$ and strictly increasing for $v> \underline{v}$. Bidder $0$'s strategy is to bribe and submit a bid only if his value $v_0 > t_1 + r$ where $t_1$ is the observed tip.

\begin{proposition}\label{prop:eqbm1r}
The following constitutes an equilibrium of the game with 2 bidders, i.e., $N= \{0,1\}$, when the seller announces a second-price auction with a reserve price $r>0$: 
\begin{itemize}
    \item Bidder 1 submits a truthful bid, and his tipping strategy as function of his value $v_1$ is given by $t_1(v_1)= v_1/2 -r$ whenever $v_1 > 2r$ and $0$ otherwise.
    \item Bidder $0$'s strategy as function of the observed tip $t_1$ and his value $v_0$ is given by 
    \begin{equation*}
    \sigma_0(t_1,v_0) = \begin{cases}
        \text{bribe} & t_1 + r\leq v_0\\
        \text{don't bribe} & \text{o.w.} \\
    \end{cases}
\end{equation*}
where bribe is shorthand for paying $t_1$ to the proposer in exchange for omitting bidder 1's transaction. Further bidder $0$ submits a non-zero bid in the auction if and only if he bribes the proposer.
\item The proposer accepts bidder $0$'s bribe whenever it is offered and omits bidder $1$'s bid, otherwise the proposer includes both bids. 
\end{itemize}
\end{proposition}

Note that while the seller does receive positive revenue in this case, they do not realize any ``benefit'' from running an auction relative to  posting a ``buy it now'' price of $r$. We formalize this in the following corollary:
\begin{corollary}
    Suppose buyer values are i.i.d. $ \text{Uniform}[0,1]$. Assuming $r \leq \frac12$, Bidder $0$ wins the object with ex-ante probability $(1-r)(1- (\frac12 -r)^2)$, while bidder $1$ wins the object with ex-ante probability $(1-2r)(\frac{r}{2} + \frac14)$. The expected revenue of the seller is $r (1-r^2)$, which is the same as the revenue for posting a ``buy it now'' price of $r$. The proposer makes an expected revenue of $\frac14 (1-2r)^2$.  
\end{corollary}

Again bidder $0$ has a strong advantage in this auction. Tips are no longer perfectly revealing, since a non-empty interval of bidder values tip $0$, but remain weakly monotone in bid and perfectly revealing when strictly positive.

\subsection{Three or more bidders}

\begin{figure}%
    \centering
    \subfloat{{\includegraphics[width=6.5cm]{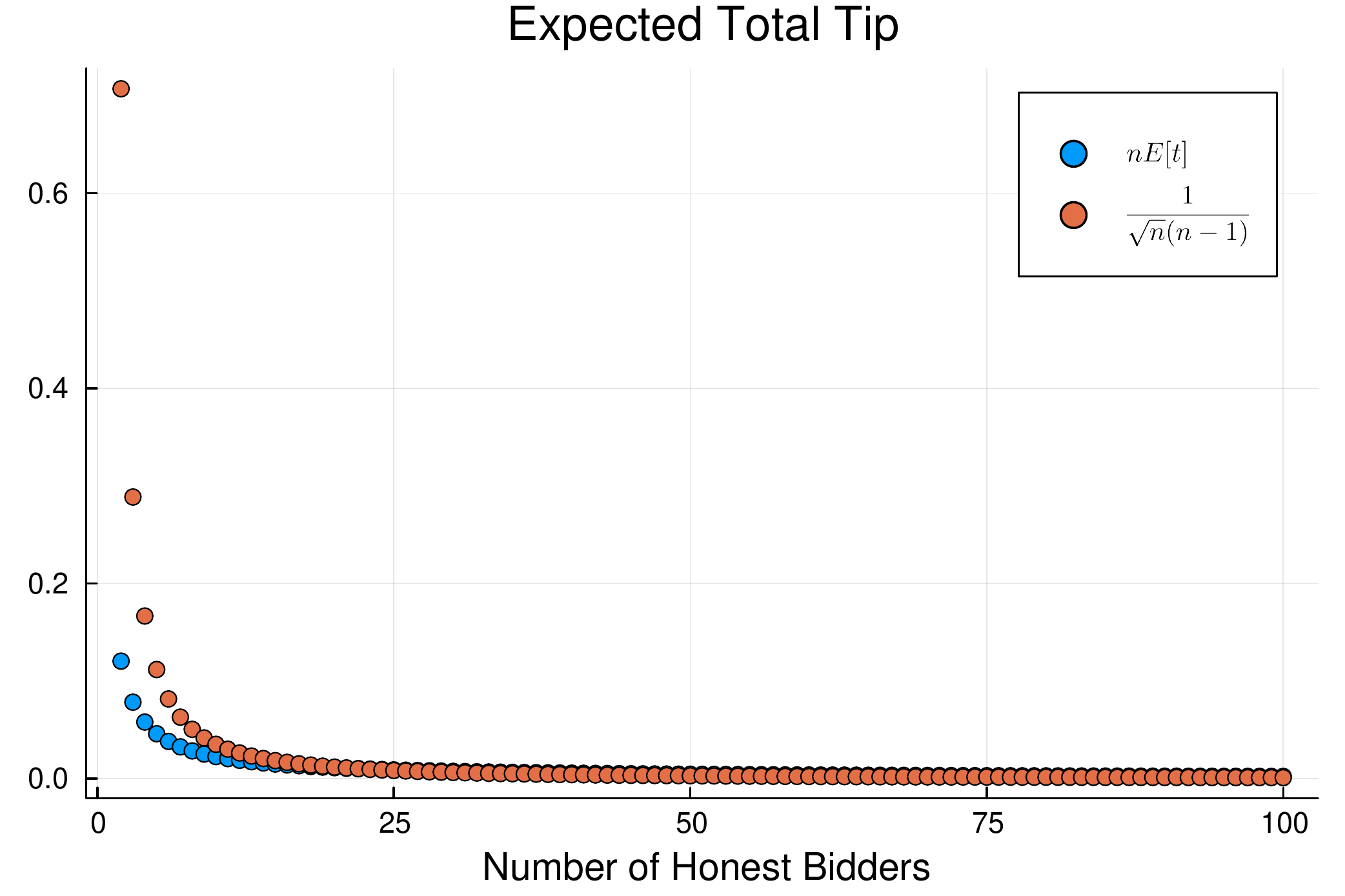}}}%
    \subfloat{{\includegraphics[width=6.5cm]{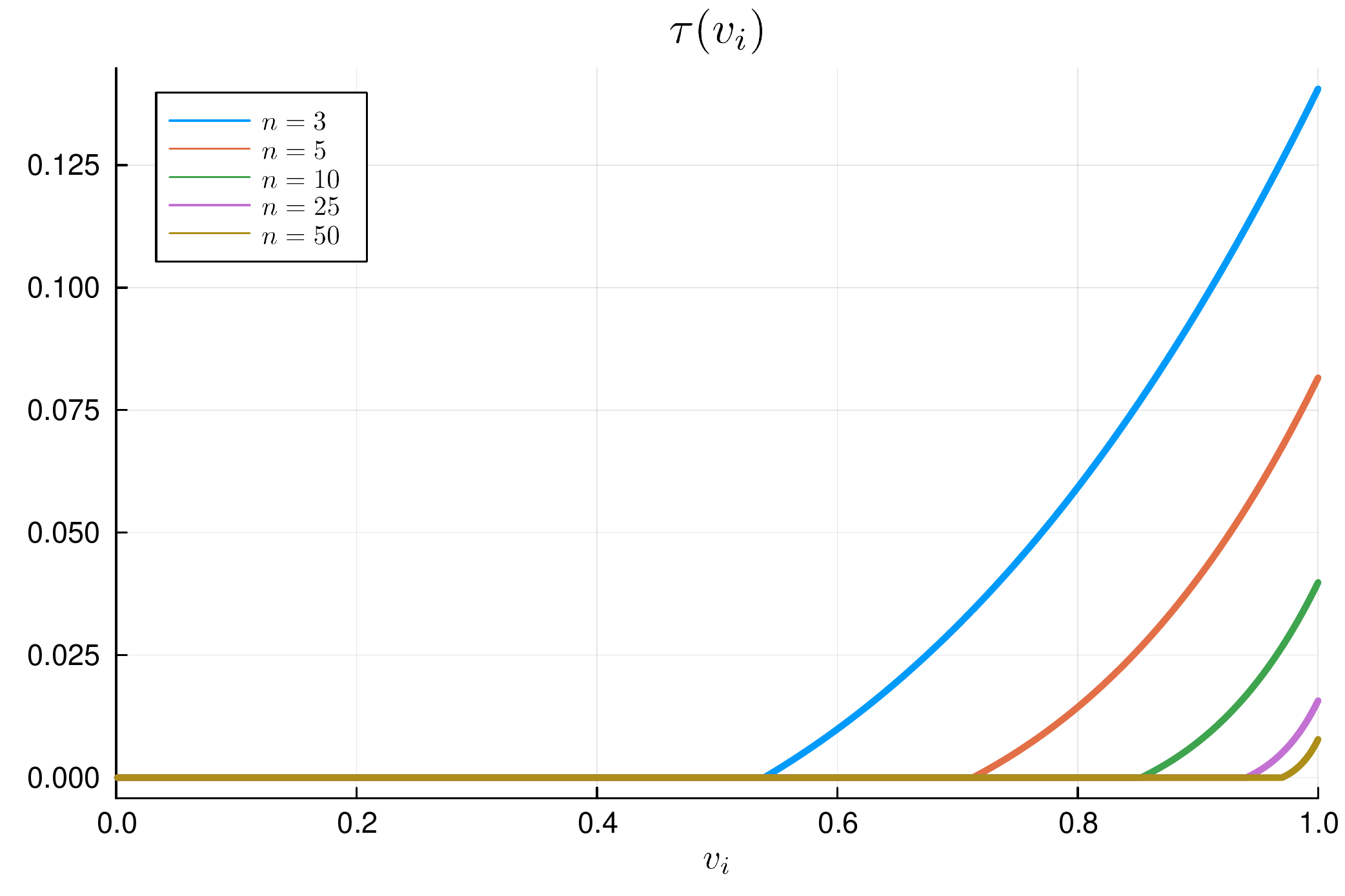} }}%
    \caption{Expected Total Tips and Tipping functions for $F = F_0 \sim \text{Uniform}[0,1]$}%
    \label{fig:Uniform_fig}%
\end{figure}

We now turn to the case where $n\geq 2,$ i.e., $N = \{0,1,\ldots,n\}$. 

At first the problem of finding an equilibrium may appear intractable since bidder $0$'s best-response problem is itself complicated: there are $2^n$ possible subsets of $\{1,\ldots,n\}$ and the problem of finding the best subset to buy out is therefore non-trivial. The tipping function of bidders $\{1,\ldots,n\}$ needs to be a best response to this (accounting for how changing their tip changes the probability that they are censored given a distribution of other tips). Nevertheless, there is an easy to describe equilibrium. To construct it, the following lemma is useful:

\begin{lemma}\label{lem:tips}
Suppose that the tipping strategy of buyers $1 \ldots n$ is such that $t(v) \leq v/n$. Then we have that the best response for bidder $0$, as a function of his own value $v_0$ and observed vector of tips $\mathbf{t} = (t_1 \ldots t_n)$ can be described as:
\begin{equation*}
    \sigma_0 (v_0,t_1,\dots,t_n) = \begin{cases}
    \text{bribe}, & \sum_{i=1}^n t_i \leq v_0\\
    \text{don't bribe}, & \sum_{i=1}^n t_i > v_0 
    \end{cases}
\end{equation*}
where bribe is shorthand for paying the proposer $\sum_i t_i$ in exchange for omitting all of bidder 1 through n's transactions. 
\end{lemma}
\begin{proof}
To see this note that: \begin{equation*}
    \sum_{i=1}^n t(\theta_i) \leq \sum_{i=1}^n \frac{\theta_i}{n} \leq \sum_{i=1}^n \frac{\max(\theta_1,\theta_2,\dots,\theta_n)}{n} = \max(\theta_1,\theta_2,\dots,\theta_n).
\end{equation*}
and therefore bribing the proposer and buying out all the bids (and therefore winning the object for free in the auction) is more profitable than buying out any subset of the bids and possibly losing the auction or having to pay more than the bribes for that subset would have cost. 
\end{proof}

This lemma is useful because, in equilibrium, the tipping strategy of bidders $1$ through $n$ will satisfy this property. Bidder $0$'s strategy is therefore straightforward (analogous to the case $N=\{0,1\}$). We are now in a position to describe the equilibrium in this game.
\begin{proposition}\label{prop:eqbmn}
The following constitutes an equilibrium of the game with $n+1$ bidders, $N= \{0,1,\ldots,n\}$, and buyer values drawn i.i.d. from $\text{Uniform}[0,1]$ when the seller announces a second-price auction with a reserve price $r=0$: 
\begin{itemize}
    \item Bidders 1 through $n$ submit truthful bids, and their tipping strategy as function of their value $v_i$ is given by: 
    \begin{align} \label{eq:tipn}
    &t(v) = \begin{cases}
        0 & v < \underline{v},\\
        \frac{1}{2n} \left({v^n}-\underline{v}^n \right) & \text{o.w. }
    \end{cases}
\intertext{where $\underline{v}$ solves} 
&(n+1) \frac{\underline{v}^n}{n(n-1)} - \frac{\underline{v}^{n+1}}{(n+1)} -\frac{1}{n(n+1)} =0.\label{eq:underlinev}
\end{align}
    \item Bidder $0$'s strategy as function of the observed tips $t_1,\ldots,t_n$ and their value $v_0$ is given by 
    \begin{equation*}
     \sigma_0 (v_0,t_1,\dots,t_n) = \begin{cases}
    \text{bribe}, & \sum_{i=1}^n t_i \leq v_0\\
    \text{don't bribe}, & \sum_{i=1}^n t_i > v_0 
    \end{cases}
\end{equation*}
where bribe is shorthand for paying $\sum_i t_i$ to the proposer in exchange for omitting bidder 1 through $n$'s transactions. Further bidder $0$ submits a truthful nonzero bid in the auction if and only if he bribes the proposer.
\item The proposer accepts bidder $0$'s bribe whenever it is offered and omits the other bids, otherwise the proposer includes all bids. 
\end{itemize}
\end{proposition}

It is easy to see that \eqref{eq:underlinev} has exactly 1 root in $[0,1]$ by observing that the left hand side is increasing in $v$ on $[0,1]$, and evaluates to a negative number at $v =0$ and a positive number for $\theta = 1$. Unfortunately, we cannot analytically derive these roots for arbitrary $n$ since polynomials of order $\geq 5$ do not have explicit roots (and even for $n=2,3$ these are not particularly nice); however, we can use a zero finding algorithm to compute these numerically. The results for the uniform case are presented in Figure \ref{fig:Uniform_fig}. 
\newcommand{\utheta}{\underline{v}}

Analytically, we can bound how this root varies with $n$. Note that the expected total tip is $n \mathbb{E}[t(v)]$ and substituting in \eqref{eq:tipn} and simplifying via \eqref{eq:underlinev}, we have that the expected total tip  $= \frac{\utheta^n}{n-1}$. The following proposition describes the asymptotic behavior of the expected total tip:
\begin{proposition}\label{prop:underlinev}
    Let $\utheta(n)$ describe the solution to \eqref{eq:underlinev} as a function of $n$. There exists $\underline{n}$ large enough such that for $n> \underline{n}$, we have $ \frac1n \leq \utheta(n) ^n  \leq \frac{1}{\sqrt{n}}$.  
\end{proposition}

Proposition \ref{prop:underlinev} is particularly useful when you consider the fact that for large $n$, by the law of large numbers, the total tip concentrates around the expected tip with high probability (buyer values are i.i.d. bounded random variables). Furthermore, by Proposition \ref{prop:underlinev}, this is decreasing at a rate at least $1/n\sqrt{n}:$ as $n$ grows, and individual bidders are willing to tip less. To see why tipping is only profitable when it leads to the bid not being censored \emph{and} winning the auction, but increasing the tip increases the probability that all bids are not censored. In short, tips have public goods type properties. Indeed, the rate of tipping shrinks fast enough so that the total tip is also decreasing. Therefore bidder $0$ wins the auction with increasing probability in $n$, asymptotically tending to $1$. Note that the seller only receives revenue when there is more than one bidder in the auction (or more generally, in the auction with a reserve price $r$, makes revenue larger than the reserve price)--- and this happens with vanishing probability as $n$ grows large.

\begin{proposition}\label{prop:eqbmrn}
    As $n$ grows large, the expected revenue of the auction with reserve price $r$ reduces asymptotically to the expected revenue of a published price of $r$. 
\end{proposition}

\subsection{General Distributions}

\begin{figure}%
    \centering
    \subfloat{{\includegraphics[width=6.5cm]{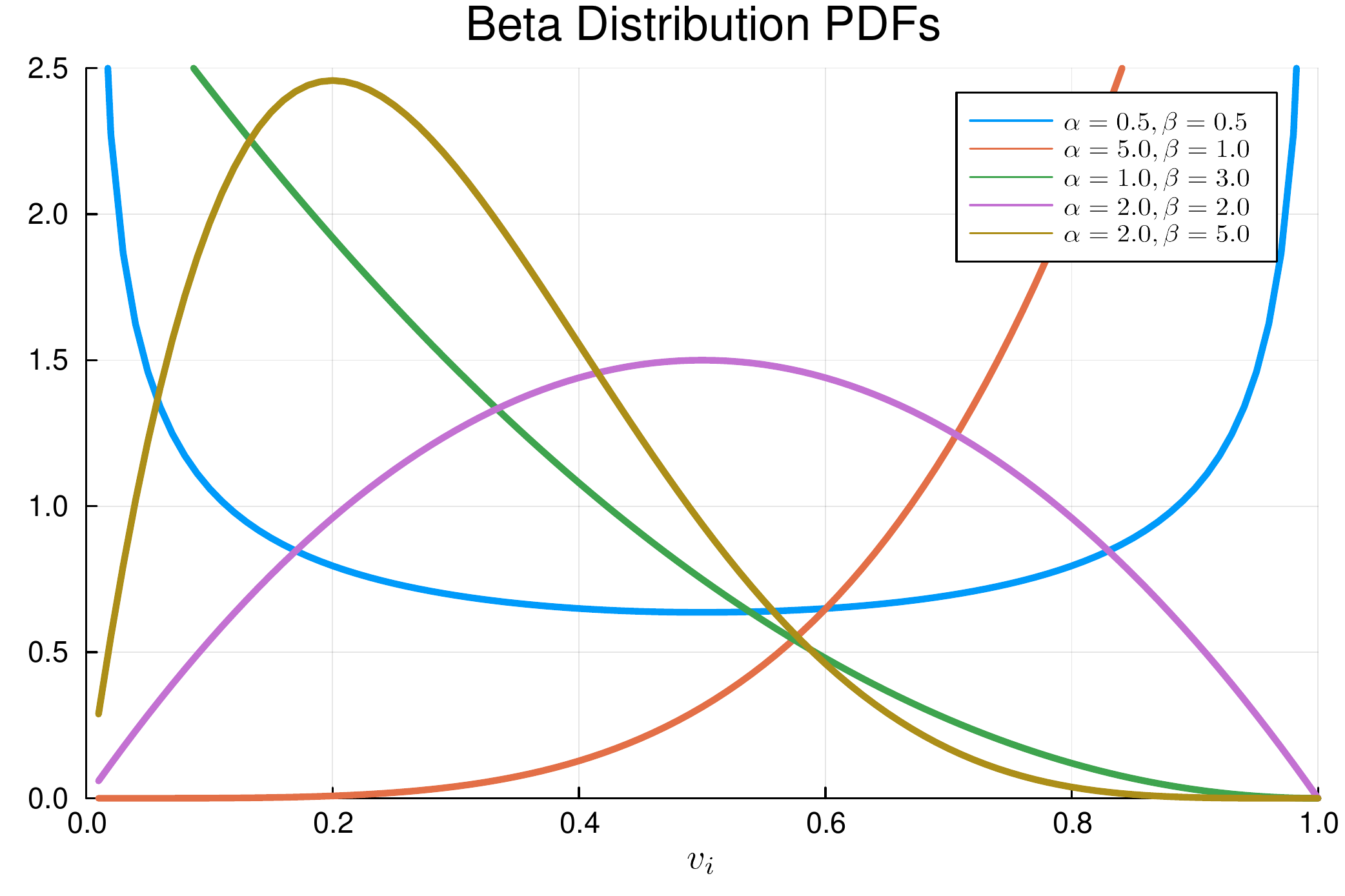}}}%
    \subfloat{{\includegraphics[width=6.5cm]{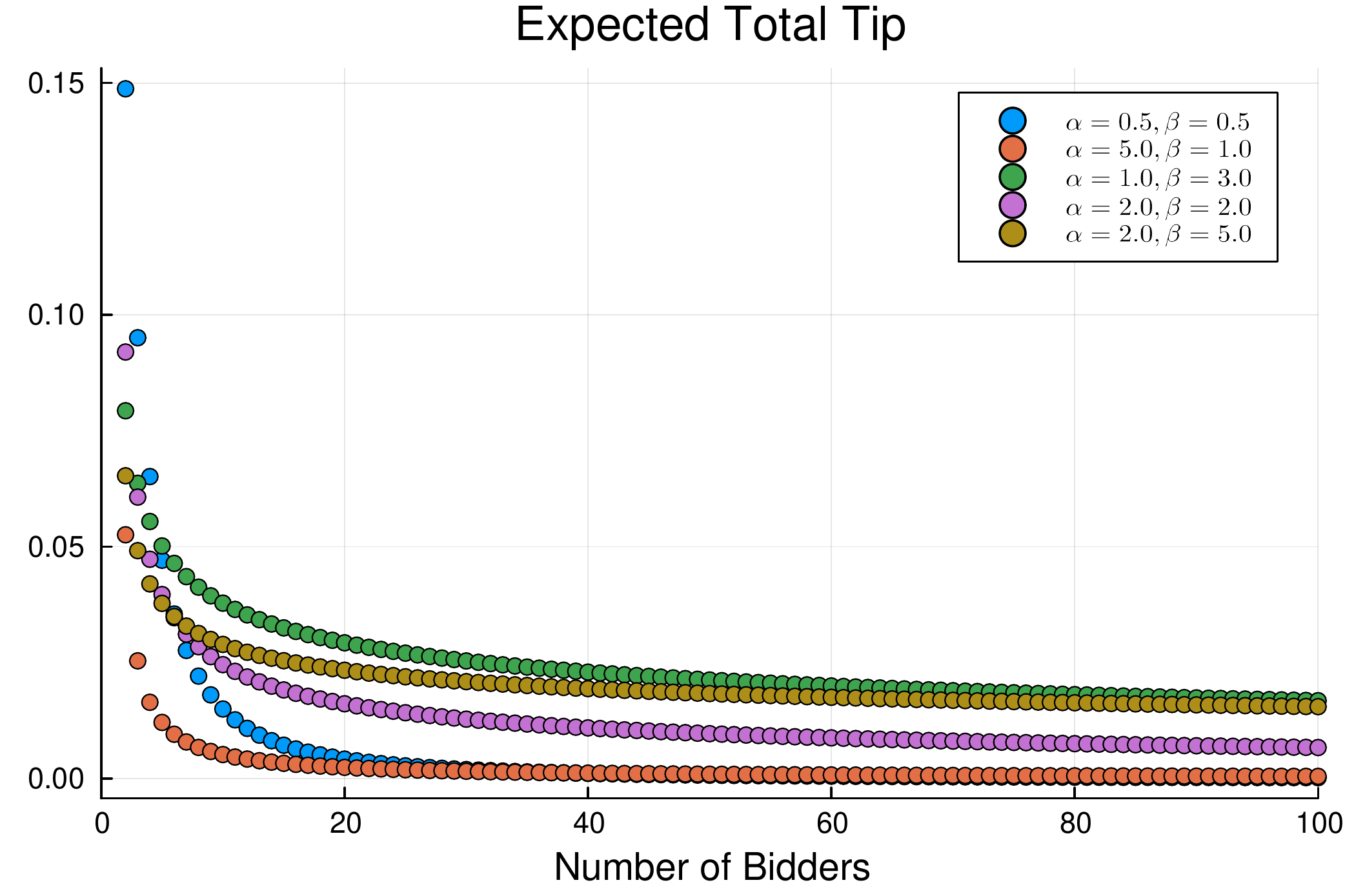} }}%
    \caption{Expected Total Tips for $F_0 \sim \text{Uniform}[0,1]$, $F \sim \text{Beta}(\alpha,\beta)$}%
    \label{fig:example}%
\end{figure}

It is straightforward to generalize our results to the case where bidders $\{1,\ldots, n\}$ are distributed according to some general distribution with density $f$ and CDF $F$ on $[0,1]$.

\begin{assumption}\label{ass:1}
We assume throughout that $F$ satisfies $\int_0^v F^{n-1}(\theta) d\theta \leq \frac{v}{n}$ for all $v \in [0,1]$.
\end{assumption}

Formally, we have the following proposition:

\begin{proposition}\label{prop:eqbmngen}
The following constitutes an equilibrium of the game with $n+1$ bidders, i.e. $N= \{0,1,\ldots,n\}$, when the seller announces a second-price auction with a reserve price $r=0$, bidder $0$ has a value $\text{Uniform}[0,1]$ and bidders $1$ through $n$ have values distributed i.i.d. according to a distribution with density $f$ and CDF $F$ satisfying Assumption \ref{ass:1}:
\begin{itemize}
    \item Bidders 1 through $n$ submit truthful bids, and their tipping strategy as function of their value $v_i$ is given by: 
    \begin{align} \label{eq:tipngen}
    &t(v) = \begin{cases}
        0 & v < \underline{v},\\
        \frac{1}{2} \int_{\utheta}^v F^{n-1}(\theta) d\theta  & \text{o.w. }
    \end{cases}
\intertext{where $\underline{v}$ solves} 
   &\int_0^1 F^{n-1}(\theta) d\theta - \int_{\utheta}^1 F^n(\theta) d\theta =  \frac{n+1}{n-1} \int_0^{\utheta} F^{n-1}(\theta) d\theta.
\end{align}
    \item Bidder $0$'s strategy as function of the observed tips $t_1,\ldots,t_n$ and their value $v_0$ is given by 
    \begin{equation*}
     \sigma_0 (v_0,t_1,\dots,t_n) = \begin{cases}
    \text{bribe}, & \sum_{i=1}^n t_i \leq v_0\\
    \text{don't bribe}, & \sum_{i=1}^n t_i > v_0 
    \end{cases}
\end{equation*}
where bribe is shorthand for paying $\sum_i t_i$ to the proposer in exchange for omitting bidder 1 through $n$'s transactions. Further bidder $0$ submits a true nonzero bid in the auction if and only if he bribes the proposer.
\item The proposer accepts bidder $0$'s bribe whenever it is offered and omits the other bids, otherwise the proposer includes all bids. 
\end{itemize}
\end{proposition}

This proposition allows us to numerically solve for tipping behavior in this auction. Using the flexible Beta distribution for a range of parameters, we compute $n\mathbb{E}[t]$ as a function of $n$ in Figure \ref{fig:example}.

Analytical results for the case where bidder $0$'s value is distributed non-uniformly appear out of reach. To see why---for bidders $1$ through $n$, part of the payoff of increasing their tip is increasing the probability that bidder $0$ chooses not to buy them out. When bidder $0$'s value is distributed uniformly, the increase in probability is constant and independent of others' tips (which from the perspective of the bidder is a random variable). This simplifies the optimality condition and makes it analytically tractable. Nevertheless, the intuition above suggests that our qualitative results (seller expected revenue drops to close to the posted price, tips do not offer much ``protection'' due to the public goods nature of tips suggests, bidder $0$ has a strong advantage in the auction) carry over to this case as well.

\section{Restoring Censorship Resistance}
We now discuss possible design choices to provide additional censorship resistance, so that an auction can be run on chain with the desired results. 

\subsection{Auction over Multiple Blocks}
We now investigate whether running the auction over multiple blocks restores the desired behavior. Formally, recall that using multiple blocks as the underlying public bulletin board had a higher censorship resistance (Example \ref{ex:multiple}) than in the case of a single block (Example \ref{ex:single}).  

Formally we consider the following dynamic game corresponding to an auction being run over $m$ blocks. We assume that each of these blocks is produced by an independent proposer.\footnote{In practice, an majority of blocks on major blockchains is produced by one of a small oligopoly of proposers. We discuss the implications of this in the sequel.}
\begin{enumerate}
    \item Period $0$: Bidders learn their values $v_i$.  Bidders $1, \ldots, n$ each submit simultaneously a private bid $b_i$ and a public tip $t_i$. 
    \item Period $j$ for $j$ in $1$ to $k$: Bidder $0$ observes which bids from $1, \ldots n$ have not been included in a block in periods 1 to $j-1$.  They offer Proposer $j$ a take-it-or-leave-it-offer of a subset $S_j$ of the unincluded bids and a payment $p_j$ to exclude that subset.  The proposer $j$ observes the tips and the offer from Bidder $0$ and decides which bids if any to include. 
    \item Period $m+1$: The seller's auction is run on blocks produced in periods $1$ to $m$. 
\end{enumerate}
If a transaction is included in period $j$, it is removed from the set of bids in the mempool, its tip is attributed to proposer $j$ where $j$ is the block it was included in, it is included in the auction, and it cannot be included in subsequent blocks. 

Note that the game we describe is a natural extension of the previous game to the case of an auction over $m>1$ blocks. In particular, it reduces to the original game for the case of $m=1$. We consider the same refinement as before, which applies to this game in a similar fashion. 

Note that there are two possible sources of additional security in this auction. The first is mechanical: in order to censor a transaction, intuitively, bidder $0$ has to bribe $m$ proposers, which is more expensive for a given tip. The second is that the marginal returns to a tip have increased (increasing a tip by $q$ increases the cost to censor for bidder $0$ by $m q$, and decreases the probability they can afford it correspondingly). 

In our results, we show that the latter effect is null. In particular, we show that for $m<n,$ the tipping behavior of bidders $1,\ldots n$ stays unaffected. 

Formally, we have the following result: 
\begin{proposition}\label{prop:eqbmnk} Suppose buyers $1$ to $n$ have values drawn i.i.d. $U[0,1]$ and buyer $0$ has value drawn i.i.d. $U[0,\kappa]$ for $\kappa >m$. The following constitutes an equilibrium of the game with $n+1$ bidders, i.e. $N= \{0,1,\ldots,n\}$, when the seller announces a second-price auction with a reserve price $r=0$ to be run over $m$ blocks for $m<n$:  Bidders 1 through $n$ and the proposers have the same strategy as in Proposition \ref{prop:eqbmn}.

Bidder $0$'s strategy as function of the observed tips $t_1,\ldots,t_n$ and his value $v_0$ is given by 
    \begin{equation*}
     \sigma_0 (v_0,t_1,\dots,t_n) = \begin{cases}
    \text{bribe}, & m \sum_{i=1}^n t_i \leq v_0,\\
    \text{don't bribe}, & \text{ otherwise,}
    \end{cases}
\end{equation*}
\end{proposition}

Before we proceed, we comment on the assumption that bidder $0$'s
 value is distributed $U[0,\kappa]$. Suppose bidder $0$ is distributed $U[0,1]$, but bidders $1$ to $n$ tip as in Proposition \ref{prop:eqbmn}. Note that with positive probability the total tip could exceed $1$ when $m>1$. Therefore, given bidders 2 to $n$ follow the tipping strategy of Proposition \ref{prop:eqbmn},  bidder $1$ will have incentives to shade their tip relative to $t(\cdot)$. After all, the marginal value of a tips depends on the how they increase the probability of not being censored. From the Proof of \ref{prop:eqbmn}, in the case of $m=1$, increasing one's tip on the margin always increases the probability that the bids are not censored, because the total tip is strictly smaller than $1$ with probability $1$ on path. Intuitively, therefore if bidder $0$ is distributed $U[0,1]$, and the auction is conducted over $m>1$ blocks, the equilibrium tipping strategy for bidders $1$ to $n$ is weakly lower than the corresponding tipping strategy for $m=1$ (Proposition \ref{prop:eqbmn}). This can be shown numerically, but is out of reach analytically. 

Note that we had already shown that expected total tip of $n$ bidders was smaller than $1/n^{3/2}$. Therefore, the probability that bidder $0$ does not censor the remaining bids collapses to $0$ as $n$ grows large, as long as $m$ grows sublinearly with $n$. To see that we had already shown that expected total tip of $n$ bidders was smaller than $1/n^{3/2}$. Therefore $m$ times this for $m <n$ still grows smaller than $1/\sqrt{n}$. 

Put differently, guaranteeing the auction outcome is ``as desired'' requires $m> n$. This comes with its own costs: for example, the auction would have to remain open for a relatively long time which may be undesirable, particularly for financial applications.

\subsection{Multiple Concurrent Block Proposers}
Depending on the number of bidders and the time constraints inherent to the specific auction application, it may not be feasible to hold the auction for long enough to achieve the desired censorship resistance level. 

A different solution we now consider would be to allow more than one proposer within a single slot. Formally, we now consider $k$ concurrent block proposers (by analogy to our previous section where we considered $k$ sequential block producers). The seller announces an auction which will execute within the single slot, i.e. the bids included on at least one of the $k$ concurrent produced blocks will be included in the auction. 

In view of the concurrency, we allow bidders to submit conditional tips, which depend on the number of proposers who include the transaction.  For simplicity, we consider a \textit{twin tip}, i.e., each bidder submits a conditional tip of the form $(t,T)$, where $T$ is paid if only a single proposer includes bidder 1's transaction and $t$ is paid if more than one proposer includes the transaction.

\begin{observation}
With $k$ concurrent proposers and conditional tipping, the censorship resistance of a conditional tip $(t,T)$ is straightforwardly verified as:
\begin{equation}
    \varphi (t,T) = kT. 
\end{equation}
\end{observation}

It is important to note that the conditional tip disentangles the cost of inclusion (for the transacting party) from the cost of censoring, i.e. if $T \gg t$, then the censorship resistance is $kT$ which is much larger than the cost of inclusion, $kt$. 

After the honest bidder observes $v_1$ and submits his private bid $b_1$ and public tip $(t,T)$, the bribing bidder submits a bribe to each proposer. Formally, we first consider the following game:
\begin{enumerate}
    \item Seller announces second-price sealed-bid auction with reserve price $r$ to be conducted over a single slot. 
    \item Buyers learn their values $v_i \sim F$. 
    \item Buyers $1, \ldots, n$ each submit simultaneously a private bid $b_i$ and a public tip $t_i, T_i$. 
    \item Buyer $0$ observes all the other tips $t_i, T_i$ and simultaneously offers each proposer a take-it-or-leave-it-offer of a subset $S \subseteq \{1,\ldots, n\}$ of bidders and a payment $p$ to exclude that subset's bids. Bidder $0$ also submits his own bid $b_0$.
    \item Each Proposer simultaneously accepts or rejects bidder $0$'s offer and  constructs the block accordingly i.e., either containing bids of set $N \setminus S$ or $N$. 
    \item The auction is computed based on the union of the bids included in all blocks, tips are paid based on the inclusion behavior. 
\end{enumerate}

As before we focus our attention on equilibria where Bidders $\{1,\ldots,n \}$ bid truthfully in the auction. Note that each proposer, in choosing whether to censor a transaction, needs to reason about the behavior of other proposers since that potentially affects their tip if they include the transaction.

\begin{proposition} \label{prop:multproposer}
Consider the Multiple Concurrent Block Proposer game, with m proposers and $n=1$ honest bidder, i.e., $N=\{0,1\}$, with bidder $i$'s value drawn from a distribution with CDF $F_i$ and PDF $f_i$. 

This game has an equilibrium where the outcome of the auction is the same as a standard second price auction without a censorship step and where the expected tip by each bidder to each proposer is \underline{t}. 

In particular, bidders $1$'s tipping strategy in equilibrium is given by:
\begin{align}
 t_1(v_1) = 0,\hspace{2cm} T_1(v_1) = 1
\end{align}
Bidder $0$'s offer strategy to the proposer based on their own value $v_0$ and the observed tips $(t,T)$ is  
\begin{align}
    z_0(t,T,v_0) = \begin{cases}
            0 & C(v_0) < mT\\
            T & C(v_0) \geq  mT.
    \end{cases}
\end{align}

Here $C(v_0)$ is buyer $0$'s net value to censoring bidder $1$'s bid (i.e., the difference their profit from censoring the competing bid and winning in the auction for free ($v_0$); and their expected surplus from competing with bidder $1$ in the auction).   
Finally, the proposer's strategies are to censor transactions with the following probabilities:
\begin{align}
    p(z,t,T) = \begin{cases}
        0 & z < t\\
        \left(\frac{z-t}{T-t}\right)^{\frac{1}{m-1}} & t \leq z < T\\
        1 & z \geq T
    \end{cases}
\end{align}
\end{proposition}

The above proposition may admittedly appear a little dense. We provide the following corollary regarding (on-path) equilibrium behavior. 
\begin{corollary}
    Consider the equilibrium proposed in Proposition \ref{prop:multproposer} for any $m \geq 2$. On path, bidder $0$ does not bribe the proposers and instead competes in the second-price auction. All bidders pay $0$ in tips on path. Equilibrium tips do not reveal bids.  
\end{corollary}

Further, a careful study of the proof of Proposition \ref{prop:multproposer} shows that the Corollary continues to hold even when $n>1$. Therefore even $2$ concurrent block proposers restore the ``desired'' outcome relative to a single block proposer system. This is partly driven by concurrency of block proposers which removes the ``monopoly'' that they have over transaction inclusion, and partly by the conditional tip. The conditional tip allows bidder $1$ to get security via a high-tip offer conditional on inclusion by only a single proposer. This high tip offer makes it very expensive for bidder $0$ to attempt censor bidder $1$'s bid since they would need to pay $m$ times the high tip to censor the bid by all $m$ producers. Therefore, no bribe is offered. Further, this tip never needs to be paid, since both proposers find it weakly dominant to include the bid and pick up the low tip.\footnote{Note that $t = 0$ supports alternative asymmetric equilibria in the inclusion subgame where only a single proposer includes the transaction. In the broader game, these would correspond to the equilibria in the single proposer cases where $T$ substitutes for the old single dimensional tip. However, when $t$ is bounded away from 0, these equilibria disappear, since it is now strictly dominant for the proposer to include the bid in the inclusion subgame.} As an aside, note that this also restores equilibrium bid privacy, since tips no longer reveal bids. 

Finally, we should note that the conditional tipping logic we identify could also be applied to an auction over multiple blocks, achieving more censorship resistance than previously. However, it doesn't 
\section{Discussion}
Our results suggest that single proposer blockchains are not ideal for holding time sensitive auctions when the number of potential bidders is large. In our results, collusion arrangements are extremely profitable for the colluding bidder but only marginally profitable for the proposer. However, this is because we restrict the model to have one potential colluding bidder who can bribe the proposer. In reality, there are many possible colluding bidders, and only one proposer in each slot, the proposer could end up charging for the right to collude and extract a significant portion of the value that the colluding bidder gains from the arrangement. In fact, the predominant block building system, MEV-boost, can be thought of as a direct channel through which the proposer can sell the right to censor transactions to the highest bidder. This suggests that one driver of MEV is the proposer's right to determine inclusion. Previous work has focused on a different source: the proposer's right to order transactions within a block. From the position that proposer ordering power is the source, order agnostic mechanisms should solve MEV. But if censorship power is the source, these order agnostic mechanisms, including the second price auction we study here, could be just as susceptible to value extraction.

Another proposed source for MEV is the public nature of transactions in the mempool. The argument is, transactions in the mempool are sitting ducks, waiting to be front-ran.  It follows that, if transactions are encrypted while in the mempool, they will be less susceptible to MEV. But our results demonstrate that even when bids in the auction are encrypted, public tips may reveal private bids. 

As the previous paragraph suggests, cryptographic solutions do not immediately solve the censorship problem unless they also appropriately modify consensus to bind the proposer (cf our multiple concurrent proposer suggestion): for example consider an MPC based approach. One way to implement it would be to use an integer comparison MPC protocol to compare the private tips and choose the top k. Then, if the proposer chose not to include those top k, they could be penalized, or the block could be considered invalid by consensus. Note that this requires changing consensus as well as the mempool, and it requires multiple validators to participate in MPC, each of whom can add transactions to the process. In that way it is similar to the multiple concurrent block proposer proposer scheme we describe. However, in this scheme, the proposer also retains the option to accept a bribe rather than use the MPC protocol, i.e., they may accept a bribe that exceeds the expected value of the top k tips.

Implemented purely as a private mempool solution without changing consensus, the main impediment is enforcing the proposer’s commitment to this block inclusion strategy: i.e., MPC allows us to compute the k transactions with the highest tips in a private fashion, but does not bind the proposer to include these. The proposer remains free to censor any subset of these transactions if it is profitable for them to do so.

\subsection{Multiple Concurrent Proposers in the Real world}
There are projects trying to implement multiple concurrent proposers by appropriately modifying the Tendermint protocol.\footnote{See, e.g., \url{https://blog.duality.xyz/introducing-multiplicity/}.} Representatives from at least one major chain, have also mentioned concurrency as a goal going forward, in order to scale throughput and decrease latency from the user to the nearest proposer within a given slot.\footnote{See, e.g., \url{https://blog.chain.link/execution-and-parallelism-for-dag-based-bft-consensus/}.} The co-founder, and CEO of Solana even mentioned MEV resistance as a motivation for the desirability of multiple concurrent block proposers.\footnote{See \url{https://twitter.com/aeyakovenko/status/1584676110948012032}.} Our results suggest that proposer rents arising from the proposer’s temporary monopoly on inclusion shrink sharply under multiple concurrent proposers using our conditional tip logic. To see the order of magnitude, note that total payments by MEV-boost to proposers on Ethereum totaled \$245 million since Sept 2022 (the merge), see https://dune.com/arcana/mev-boost). This is quite large in comparison to the total value created by MEV and contributes to high transaction fees for users. That said multiple proposer protocols are not a panacea. Beyond the engineering difficulties themselves, we see a few other impediments. Firstly, ordering---with a single proposer, a canonical ordering of transactions within the block is simply the order that the leader put them in. Once we add multiple proposers, the ordering becomes less clear. Secondly, multiple proposers may be more susceptible to Distributed Denial of Service (DDOS) attacks: allowing every validator to submit transactions opens the network up to a DDOS attack where adversaries include many transactions to slow down the chain. Finally there are issues of redundancy and state bloat: the multiple proposer approach potentially leads to considerable redundancy within the slot (the same transaction may be included by many different proposers). 
\subsection{Future Directions}
From a theoretical perspective, this work leaves open questions of how censorship in on chain auctions might effect the equilibria of auctions with different assumptions about bidder valuations. For example a natural extension to our results would be to consider honest bidders with interdependent or common values. This may be a better model for on chain order flow auctions and collateral liquidation auctions. 

Our results are based on the assumption that a single bidder has been selected as the colluding bidder in advance, but an assumption that would track more closely to the realities of MEV-boost would be to have the right to collude be auctioned off after bids have been submitted. If the right to collude is auctioned off before the bids are submitted, then our results would still hold, since our result would simply be a subgame of the larger game being considered, and the result would be that proposers end up with a larger share of their monopoly rents; however, when the right to collude is auctioned off after transactions have been submitted, and therefore after bidders discover their types, the players who are willing to pay to collude are  more often those who value the item more. This could warp the equilibrium slightly. 

Outside of mechanism design, this work provides a strong theoretical justification for investigating multiple concurrent block proposer based consensus frameworks as a tool for MEV mitigation. Specifically, we identify conditional tipping as a powerful tool to combat censorship in situations where there are more than one block proposer. 

Indeed, strong censorship-resistance is beneficial to other systemically important smart contracts. Oracle feeds require timely inclusion guarantees to be useful, and censorship in order to manipulate on-chain financial markets is a serious concern. Hedging contracts such as options need to be censorship resistant so that they can be exercised, and conversely, the underwriter has an incentive to censor. Similarly, optimistic rollups rely on censorship resistance for their security, an attacker who wishes to manipulate an optimistic rollup only needs to censor the fraud proofs for a period of time. To combat this, optimistic rollups currently leave a long window for fraud proof submission before finalizing (e.g. 7 days). Stronger censorship resistance could allow them to achieve finality faster with the same security.

Another advantage is that proposer rents arising from the proposer’s temporary monopoly on inclusion shrink sharply under multiple concurrent proposers using our conditional tip logic. To see the order of magnitude, note that total payments by MEV-boost to proposers on Ethereum totaled \$245 million since Sept 2022 (the merge), see \url{https://dune.com/arcana/mev-boost}). This is quite large in comparison to the total value created by MEV and contributes to high transaction fees for users.

Another potential tool for combating censorship on chain, that we have not discussed, is a data availability layer. Instead of being submitted to a blockchain directly, bids could be submitted to a data availability layer, nodes could then compute the results of the auction based on whichever transactions were included on the data availability layer. This is similar to holding the auction directly on chain except that the nodes tasked with curating a data availability layer do not necessarily need to participate in consensus. The requirements for a data availability layer are weaker than those required for a full blockchain so it may be easier to integrate multiple proposer architecture on data availability layers than on blockchains themselves.

\newpage

\bibliographystyle{aer}
\bibliography{onchain-auctions}

\newpage

\appendix
\section{Omitted Proofs}
\subsection{Proof of Proposition \ref{prop:eqbmn}}
\begin{proof}
We first need consider the tipping strategy $t(v)$. To do this we first work out the expected utility for player 1 when all other $n-1$ ``honest'' bidders  are bidding according to $t(\cdot)$ and bidder 0 is following the strategy above, i.e., censoring all the bids whenever the total tip doesn't exceed their value: 
\begin{align}\label{eqn:indirectutil}
    U(v, t) =  &\int_0^{v} \ldots \int_0^v \left(F_0(\sum_{j=2}^n t(v_j) + t) (v - \max (v_j) )\right) dv_2 \ldots dv_n  \\&- t\mathbb{E}\left[F_0( \sum_{j=2}^n t(v_j) + t)\right].\notag
\end{align}
Here the first term is expected profit of the buyer \emph{in the auction}, and the second term is the expected cost of the tip: tip times the probability that the bid is not censored and the tip is therefore charged. Note that since the buyer $0$'s value is distributed uniformly, so $F_0( \sum_{j=2}^n t(v_j) + t) = \sum_{j=2}^n t(v_j) + t$.\footnote{Formally, this is only true as long as $\sum_{j=2}^n t(v_j) + t<1 $, but our solution will be such that $\sum_1^n t(v_i) \leq 1$ for any profile of values $i$.}
\begin{align*}
    \implies U(v, t) =  &\int_0^{v} \ldots \int_0^v \left((\sum_{j=2}^n t(v_j) + t) (v - \max (v_j) )\right) dv_2 \ldots dv_n  \\ &- t\mathbb{E}\left[(\sum_{j=2}^n t(v_j) + t)\right]. \\
\intertext{Differentiating with respect to t}
 \frac{\partial U(v,t)}{\partial t} =&\int_0^{v} \ldots \int_0^v \left( (v - \max (v_j) )\right) dv_2 \ldots dv_n  - \mathbb{E}\left[(\sum_{j=2}^n t(v_j))\right]  - 2t. \\
\intertext{Therefore optimality implies that $t(v)$ must solve:}
 \implies & \frac{v^n}{n} - \mathbb{E}\left[(\sum_{j=2}^n t(v_j))\right] - 2t(v) \leq 0.
\intertext{with the inequality binding when $t (v) >0$.}
\intertext{Let us denote $\mathbb{E}[ t(v_j)]$ by c.}
 \implies & \frac{v^n}{n} - (n-1)c - 2t =0.
\end{align*}
Therefore we have that: 
\begin{align*}
    t(v) = \begin{cases}
        0 & v < \utheta,\\
        \frac12 \left(\frac{v^n}{n} - (n-1)c \right) & \text{o.w. }
    \end{cases}
\end{align*}
where $\utheta = ((n^2-n)c)^{1/n}$. 
Finally, since $\mathbb{E}[ t(v_j)]=c$, it must be that 
\begin{align*}
    &\int_{\utheta}^1  \frac12 \left(\frac{v^n}{n} - (n-1)c \right)dv =c, \\
\implies & \frac{1}{n(n+1)}(1-\utheta^{n+1}) - (n-1) c (1-\utheta) = 2c.  
\end{align*}
Putting together we have that 
\begin{align}\label{silly}
    \frac{\utheta^{n+1}}{n(n+1)} + (n+1) c - (n-1)c \utheta = \frac{1}{n(n+1)}. 
\end{align}
Recall that $\utheta = ((n^2-n)c)^{1/n}  \implies c = \frac{\utheta^n}{n(n-1)}$. 
Substituting this into \eqref{silly} we have 
\begin{align}
 &   \frac{\utheta^{n+1}}{n(n+1)} + (n+1) \frac{\utheta^n}{n(n-1)} - (n-1) .\frac{\utheta^n}{n(n-1)}.  \utheta = \frac{1}{n(n+1)}. \notag\\
 \implies & (n+1) \frac{\utheta^n}{n(n-1)} - \frac{\utheta^{n+1}}{(n+1)} -\frac{1}{n(n+1)} =0,\label{silly2}
\end{align}
as desired.

Further by observation, we have that the assumption of Lemma \ref{lem:tips} is satisfied, i.e., $t(v) \leq \frac{v}{n}$ and therefore we have from Lemma \ref{lem:tips} that Bidder $0$'s bribing strategy is a best response to these tips. 
\end{proof}

As an aside, we note that the the first step in the proof, i.e. noting that $F_0( \sum_{j=2}^n t(v_j) + t) = \sum_{j=2}^n t(v_j) + t$ under the uniform distribution, allows to make the  $\mathbb{E}[F_0( \sum_{j=2}^n t(v_j) + t)]$ in the first expression analytically tractable. This is the reason we are unable to generalize beyond the uniform distribution for bidder $0$.

\subsection{Proof of Proposition \ref{prop:underlinev}}
First suppose $\utheta^n = 1/\sqrt{n}$. Substituting in to the left hand side of \eqref{silly2}, we have 
\begin{align*}
    \frac{n+1}{n(n-1)\sqrt n} - \frac{1}{(n+1)(\sqrt{n})^{\frac{n+1}{n}}} - \frac{1}{n(n+1)}.
\end{align*}
Note that the first term is $O(1/ n^{3/2})$, and the latter two terms are smaller. So the left hand side is always positive. 

Therefore we have that $\utheta^n \leq \frac{1}{\sqrt{n}}$ (since the left hand side of \eqref{silly2} goes from negative to positive and is 0 at the unique root). 

This in turn implies that $nc \leq  \frac{1}{(n-1)\sqrt{n}}$. 

By a similar argument we can argue that $\utheta^n > 1/n$. To see this, note that substituting into left hand side of \eqref{silly2}, we have
\begin{align*}
     &\frac{n+1}{n^2(n-1)} - \frac{1}{(n+1)(n)^{\frac{n+1}{n}}} - \frac{1}{n(n+1)},\\
     =& \frac{(n+1)^2 - n(n-1)}{n^2(n-1)(n+1)} - \frac{1}{(n+1)(n)^{\frac{n+1}{n}}},\\
     =& \frac{3n+1}{n^2(n-1)(n+1)} - \frac{1}{(n+1)(n)^{\frac{n+1}{n}}}.
\end{align*}
Note that the first term is $O(1/n^3)$, while the second term is $\approx O(1/n^2)$, so for n large this will be negative. 

\subsection{Proof of Proposition \ref{prop:eqbmrn}}
Repeating the calculation of the proof of Proposition \ref{prop:eqbmn}, we have that a bidder of value $v$' indirect utility from a tip of $t$ given all other $n-1$ bidders tip according to $t(\cdot)$ is given by 
\begin{align*}
    U(v, t) =  &\int_r^{v} \ldots \int_r^v \left(P_0(v_0 < r+ \sum_{j=2}^n t(v_j) + t) (v - \max (v_j) )\right) dv_2 \ldots dv_n  \\ &+ P_0(v_0 < r+t ) r^{n-1} (v-r) - t \mathbb{E}\left[F_0( r+ \sum_{j=2}^n t(v_j) + t)\right] t.
\intertext{Again, the buyer $0$'s value is distributed uniformly, so $F_0( r+\sum_{j=2}^n t(v_j) + t) = r+ \sum_{j=2}^n t(v_j) + t$.}
    \implies U(v, t) =  &\int_r^{v} \ldots \int_r^v \left((r+\sum_{j=2}^n t(v_j) + t) (v - \max (v_j) )\right) dv_2 \ldots dv_n  \\&+ (r+t) r^{n-1} (v-r) - t \mathbb{E}\left[(r+\sum_{j=2}^n t(v_j) + t)\right] . 
\end{align*}
Differentiating with respect to t, we have
\begin{align*}
\frac{\partial U(v,t)}{\partial t} &\int_r^{v} \ldots \int_r^v \left( (v - \max (v_j) )\right) dv_2 \ldots dv_n + r^{n-1} (v-r) - \mathbb{E}\left[r + (\sum_{j=2}^n t(v_j))\right] - 2t.
 \end{align*}
 The first two terms are just the interim expected surplus of a bidder with value $v$ in a second price auction with n bidders and reserve price $r$, which by revenue equivalence is the integral of the allocation of lower types, i.e.,
 \begin{align*}
  &\int_r^{v} \ldots \int_r^v \left( (v - \max (v_j) )\right) dv_2 \ldots dv_n + r^{n-1} (v-r) \\=& \int_r^v \theta^{n-1} d\theta,\\
 =& \frac{v^n}{n} - \frac{r^n}{n}.
\intertext{Therefore, substituting back in, optimality of $t(\cdot)$ requires that }
&\frac{v^n}{n} - \frac{r^n}{n} - r - \mathbb{E}\left[\sum_{j=2}^n t(v_j)\right] - 2t(v) \leq 0, 
\end{align*}
with equality whenever $t(v)>0$. 
Therefore we have that:
\begin{align*}
    &t(v) = \begin{cases}
             \frac12 \left(\frac{v^n}{n} - (n-1) c_r - (r+  \frac{r^n}{n}) \right) & v \geq  (n( (n-1) c_r + (r+  \frac{r^n}{n}))^{1/n}\\
             0 & \text{ otherwise}
            \end{cases}\\
\text{where }    &c_r  = \mathbb{E}[t(v)].       
 \end{align*}
 Therefore we have  $t(v) = \frac12 \left(\frac{v^n}{n} - (n-1) c_r - (r+  \frac{r^n}{n})\right) $ when $ v > \underline{v}$ and $0$ otherwise. Recall that when $r=0$, $t(v) = \frac12 \left(\frac{v^n}{n} - (n-1) c\right) $. Therefore $c_r< c$ and $ (n-1) c_r - (r+  \frac{r^n}{n}) > n-1 c.$

 Therefore the equilibrium tipping function $t(v)$ when $r>0$ are first order stochastically dominated by the tipping function when $r=0$. We already showed in Proposition \ref{prop:underlinev} that the tipping function when $r=0$ as a function of $n$  implies that total tips collapse to $0$ at a super-linear rate. Therefore that continues to be the case here, and the Proposition follows.

\subsection{Proof of Proposition \ref{prop:eqbmngen}}
Repeating the calculation of the Proof of Proposition \ref{prop:eqbmn} but assuming that the bidders are now distributed according to some distribution $F,$ an taking first order conditions, we must have that  
\begin{align*}
    \int_0^v F^{n-1}(v) dv - c - 2 t(v) \leq 0 
\end{align*}
with it binding whenever $t(v) >0$. 

So we must have:
\begin{align*}
    t(\theta) = \begin{cases}
        0 & v \leq \utheta\\
        \frac12 \left( \int_0^v F^{n-1}(v) dv - (n-1) c\right) &\text{o.w.}
    \end{cases}
\end{align*}
where $c, \utheta$ jointly solve
\begin{align*}
    &\int_0^{\utheta} F^{n-1}(v) dv = (n-1)c,\\
    &\int_0^1 t(v) dF(v) = c. 
\end{align*}
Note that the latter equation can be written as:
\begin{align*}
    &\int_{\utheta}^1 \frac12 \left( \int_0^v F^{n-1}(v) dv - (n-1) c\right) dF(v)=c\\
    \implies& \int_{\utheta}^1 \left( \int_0^v F^{n-1}(v) dv - (n-1) c\right) dF(v) = 2c,\\
    \implies & \int_{\utheta}^1 \left( \int_0^v F^{n-1}(v) dv \right) dF(v) = \left(2 + (n-1) (1-F(\utheta))\right)c,
\intertext{writing $\int_0^v F^{n-1}(v) dv = S(v)$ and doing integration by parts}
\implies & S(1) - S(\utheta) F(\utheta)  -\int_{\utheta}^1 F^n(v) dv = \left(2+ (n-1)(1-F(\utheta))\right) \frac{S(\utheta)}{n-1},\\
\implies & S(1) - \int_{\utheta}^1 F^n(v) dv =  \frac{n+1}{n-1} S(\utheta),
\end{align*}
as desired. 

Note that under Assumption $1$, \ref{ass:1}, this implies that the assumption of Lemma \ref{lem:tips} continues to be satisfied, i.e., $t(v) \leq \frac{v}{n}$ and therefore we have from Lemma \ref{lem:tips} that Bidder $0$'s bribing strategy is a best response to these tips.

\subsection{Proof of Proposition \ref{prop:eqbmnk}}
The proof follows from the Proof of Proposition \ref{prop:eqbmn}. To see this, note that if bidder $0$ is distributed  $U[0,\kappa]$, given the conjectured strategies of $1$ to $n$, we have the total tip is always smaller than the highest possible value of bidder $0$. Therefore the proof continues to hold as written: bidder $0$ being distributed $U[0,\kappa]$ changes (\ref{eqn:indirectutil}) by a multiplicative constant $\frac{1}{\kappa}$, and therefore this does not affect the optimization. 

\subsection{Proof of Proposition \ref{prop:multproposer}}
\begin{proof} 
We proceed by backward induction. First consider the subgame where tips $(t,T)$, and bribe $z$ have already been offered. Given that all other proposers are censoring with probability $p$. The expected utility for proposer $i$ as a function of censoring probability $p_i$ is given by:
\begin{align*}
    &p_i z + (1-p_i)(p^{m-1})T + (1-p_i)(1-p^{m-1})t.
\intertext{Differentiating this with respect to $p_i$ we have} 
    &z -p^{m-1}T - t + p^{m-1}t.
\intertext{Therefore for a symmetric mixed equilibrium of the subgame it must be that}
   &0  =  z - t -(T-t)p^{m-1}\\
   \implies  &p =  \left(\frac{z-t}{T-t}\right)^\frac{1}{m-1}
\intertext{Given this, the probability of successfully censoring the bid for a given bribe $z \in [t,T]$ is} 
    &p^m = \left( \frac{z-t}{T-t}\right)^\frac{m}{m-1}
\intertext{If the value of censorship to the bribing bidder is $C$ relative to the no censorship case, then the expected utility of the bribing bidder relative to the no censorship case as a function of $z$ is}
   & C p^m - mzp = C \left( \frac{z-t}{T-t} \right)^{\frac{m}{m-1}} - m z \left(\frac{z-t}{T-t}\right)^\frac{1}{m-1}
\intertext{Define $u \equiv \frac{1}{T-t}^{1/m-1}$, and $q \equiv z-t$. Substituting in to the latter expression, we have }
    &u^m C q^{\frac{m}{m-1}} - u\left(mt q^{\frac{1}{m-1}} + mq^{\frac{m}{m-1}}\right)
\end{align*}
Note that this is convex in $q$, and therefore the solution is to either tip $t$ (which effectively is the same as tip $0$), or tip $T$. Finally, note that $C \leq v_0$ by definition.

Given this, the strategy of $(t,T) = (0, 1)$ is a best response for bidder $1$ since it will result in bidder $0$ offering a bribe of $0$. 
\end{proof}

\end{document}